\documentclass[runningheads]{llncs}
\usepackage[T1]{fontenc}
\usepackage{graphicx}
\usepackage{amsmath}
\usepackage{amssymb}
\usepackage{hyperref, xcolor}
\usepackage{multirow}
\usepackage{booktabs}

\newcommand{\metodi}{ICSA}

\begin{document}
\title{Data anonymization in the presence of outliers via invariant coordinate selection}
\titlerunning{Data anonymization in the presence of outliers via ICS}
%
\author{Katariina Perkonoja\inst{1}\orcidID{0000-0002-3812-0871} \and
Joni Virta \inst{1}\orcidID{0000-0002-2150-2769}}
\authorrunning{K. Perkonoja and J. Virta}
%

\institute{Department of Mathematics and Statistics, University of Turku, Finland}
\maketitle              
\begin{abstract}

Protecting confidential data while preserving utility is particularly challenging when data sets contain outlying observations. Existing latent space anonymization methods, such as spectral anonymization (SA), rely on principal component analysis (PCA) and may therefore be vulnerable to contamination. We investigate anonymization in the presence of outliers and propose \metodi{}, a robust alternative to SA based on invariant coordinate selection (ICS). By replacing the PCA transformation with ICS, the robustness of the anonymization procedure can be regulated through the choice of scatter matrices. Alongside the methodological development, we derive a theoretical result showing that SA fails under sufficiently influential outliers. To assess the practical implications of this result, we compare the privacy--utility trade-off of \metodi{} and SA through simulation experiments under varying contamination settings and outlier severities. Our findings indicate that implementations of \metodi{} based on robust scatter matrices achieve stronger privacy protection than SA, while typically maintaining comparable, and in some cases improved, utility. We further examine the empirical performance of the proposed method using a benchmark clinical data set, where \metodi{} demonstrates superior overall privacy--utility efficiency relative to SA. These results suggest that explicitly accounting for outliers can materially improve anonymization performance and that robust latent space transformations offer a promising direction for privacy-preserving statistical data release.

\keywords{Data anonymization  \and Invariant coordinate selection \and Outlier contamination \and Latent space transformation \and Statistical disclosure control}
\end{abstract}
\section{Introduction}\label{sec:introduction}

This work investigates the anonymization of data sets that contain outlying observations. By outlier we mean a subject whose measurements differ considerably from the bulk of the data, making them particularly susceptible to identity disclosure. We work under the assumption that the outliers represent the population, meaning that simply removing them from the data is not an option and would bias the results. Letting $X$ and $X^*$ denote the original data and the anonymized data sets, respectively, this means that the anonymized data should reach a compromise between the following two competing outcomes.

\begin{enumerate}
    \item[(i)] (Utility) Outliers should be replicated faithfully enough in $X^*$ that analyses based on $X^*$ and $X$ agree to sufficient degree.
    \item[(ii)] (Privacy) Outliers should not be replicated too closely in $X^*$ to preserve the privacy of the original subjects.
\end{enumerate}

A standard strategy to masking observations is adding noise to the data, see the review in \cite{mivule2013utilizing}. This approach, while effective for hiding non-outlying data points, is, however, suboptimal in the presence of outliers as extremely strong noise would be required to mask them, leading to a large loss in utility. More suitable approaches instead redistribute the existing variability in the data to achieve the masking while preserving the dependency structure of the data \cite{calvino2017fa,naldi2018hiding,d2020privacy,perkonoja2024asymptotic}. One such method, which we take as our starting point is \textit{spectral anonymization} (SA)~\cite{lasko2009spectral}.

SA uses the principal component decomposition of $X$ to conduct randomized anonymization and its exact algorithm is described later in Section~\ref{sec:methods}. Despite its competitive performance under regular enough data \cite{kundu2019privacy}, SA suffers a serious drawback when the data contain outliers, compromising objective (ii) above. Namely, Theorem \ref{theo:motivating_result} below states that, when the data contain a strong outlier ($x_{n + 1}$ in the result), the relative distance of this outlier to its closest match in the anonymized data is, over all possible anonymizations, bounded above by a quantity that goes to zero as the severity of the outlier increases. Below, the set of all possible spectral anonymizations of a data set $X$ is denoted by $\mathcal{S}(X)$.

\begin{theorem}\label{theo:motivating_result}
    Let $X = (x_1, \ldots, x_{n + 1})' \in \mathbb{R}^{(n + 1) \times p}$ be a $p$-variate sample where $\| x_1 \|, \ldots, \| x_n \| \leq M$ for some $M > 0$. Then,
    \begin{align*}
        \max_{X^* \in \mathcal{S}(X)} \left\{ \min_{i = 1, \ldots, n + 1} \left( \frac{\| x_{n + 1} - x_i^* \|^2}{\| x_{n + 1} \|^2 } \right) \right\} = \mathcal{O}\left(  \frac{M}{\| x_{n + 1} \|} \right),
    \end{align*}
    when $M/\| x_{n + 1} \| \rightarrow 0$.
\end{theorem}

From a practical viewpoint, Theorem \ref{theo:motivating_result} indicates that spectral anonymization should not be used when there is any likelihood of an outlier in the data, as this will be faithfully reproduced in the anonymized data, compromising privacy.

As a solution to the problem raised in Theorem \ref{theo:motivating_result}, the current work proposes \textit{invariant coordinate selection anonymization} (\metodi{}), a novel anonymization method that, like SA, performs the anonymization in the latent spectral space but, contrary to SA, is resistant to the effect of outliers. We achieve this by replacing the principal component decomposition used in SA with the invariant coordinate selection (ICS) transformation \cite{tyler2009invariant}. ICS has been previously used as a pre-processing tool for, e.g., clustering \cite{fischer2017subgroup,alfons2024tandem} and outlier detection \cite{ruiz2022detecting,archimbaud2025ics}, but to the best of our knowledge this is the first time it has been combined with anoymization.

ICS maps data into a latent space similarly to PCA, but instead of maximizing variance, it selects the latent space using three tuning ``parameters'', $T(X), S_1(X), S_2(X)$, which are essentially functions of the input data. Choosing functions with specific characteristics leads to the ICS-transformation inheriting the same properties. In particular, by choosing $T(X), S_1(X), S_2(X)$ that are resistant to outliers leads to the transformation itself being unaffected by outliers as well. This allows us to perform the anonymization in an outlier-resistant latent space, after which a back-transformation to the original scale yields anonymized data with, as we show in this work, improved privacy for outlying observations and with minimal loss, and sometimes even an improvement, on utility.

The main contributions of this work are: (i) Our Theorem \ref{theo:motivating_result} formalizes the breaking-down of SA under outliers. (ii) As an outlier-resistant alternative to SA, we propose a novel anonymization method, \metodi{}, and discuss its tuning parameter selection. (iii) We empirically compare the privacy--utility trade-off of \metodi{} with that of SA from the viewpoint presented in the beginning of Section~\ref{sec:introduction} under wide simulation scenarios. (iv) Using a clinical benchmark dataset \cite{WBCD_data}, we illustrate the practical performance of \metodi{} and demonstrate that it attains a superior privacy--utility trade-off relative to SA.

The paper is organized as follows. In Section~\ref{sec:methods} we review the original SA and describe the mechanics behind \metodi{}. Simulations and a data example comparing SA and \metodi{} in both utility privacy are collected to Section~\ref{sec:results} and in Section~\ref{sec:discussion} we conclude with discussion. The source code for this work is available at \url{https://gitlab.utu.fi/kakype/icsa}.

\section{Methods}\label{sec:methods}

\subsection{Spectral anonymization}

Let $X \in (x_1, \ldots, x_n) \in \mathbb{R}^{n \times p}$ be an observed data matrix and denote the sample mean vector and sample covariance matrix of the data by $\Bar{x}$ and $S$, respectively. In SA, the eigenvector matrix $V \in \mathbb{R}^{p \times p}$ of $S$ is first computed and the $n \times p$ matrix of centered principal components is formed as $Z = (X - 1_n \Bar{x}') V$ where $1_n \in \mathbb{R}^n$ is a vector of ones. Then, $Z^* \in \mathbb{R}^{n \times p}$ is obtained by separately applying a random permutation to each column of the matrix $Z$. Spectrally anonymized data are then reconstructed as $X^* = Z^* V' + 1_n \Bar{x}'$, see \cite{lasko2009spectral} and also \cite{perkonoja2024asymptotic} for two variants of this procedure.


By performing the permutations in the spectral PCA-space, SA approximately retains the moment structure of the original data while simultaneously masking the observed individuals with the permutation operation; the cost of this approximation to asymptotic efficiency was established in \cite{perkonoja2024asymptotic}. However, as shown in Theorem \ref{theo:motivating_result} in Section \ref{sec:introduction}, the masking operation is compromised when the data contain strong outliers, prompting us to propose \metodi{}, described next.

\subsection{Invariant coordinate selection anonymization}

ICS involves selecting three tuning functions: the \textit{location vector} $T(X) \in \mathbb{R}^p$ and two symmetric and positive definite \textit{scatter matrices} $S_1(X), S_2(X) \in \mathbb{R}^{p \times p}$. Location vectors measure the central tendency of a data set and can be seen as generalizations of the sample mean vector $\Bar{x}$, whereas scatter matrices measure different aspects of variation of $X$ (classical example is the sample covariance matrix $S$). In ICS, these functions are assumed to be \textit{affine equivariant}, that is, when the data undergo a transformation $(x_1, \ldots, x_n) \mapsto (Ax_1 + b, \ldots, Ax_n + b)$ for any invertible $A \in \mathbb{R}^{p \times p}$ and any $b \in \mathbb{R}^p$, the location and scatters transform as $T(X) \mapsto A T(X) + b$, $S_1(X) \mapsto A S_1(X) A'$ and $S_2(X) \mapsto A S_2(X) A'$. Various choices for these tuning functions and their selection is described in detail in Section \ref{subsec:parameters} below.

In ICS, the data are first standardized: $X_{\mathrm{st}} = \{ X - 1_n T(X)' \} S_1^{-1/2}(X)$ where $S_1^{-1/2}(X)$ denotes the symmetric inverse square root of $S_1(X)$. Intuitively, this first step centers the data $X$ and removes all variation measured by $S_1(X)$. Next, the second scatter matrix $S_2(X_{\mathrm{st}})$ is computed on the standardized data and its eigenvector matrix $V \in \mathbb{R}^{p \times p}$ is obtained. This is used to transform the standardized data into the latent \textit{ICS-components},
\begin{align*}
    Z = X_{\mathrm{st}} V =  \{ X - 1_n T(X)' \} S_1^{-1/2}(X) V \in \mathbb{R}^{n \times p}.
\end{align*}
This final step rotates $X_{\mathrm{st}}$ such that the variation measured by $S_2(X_{\mathrm{st}})$ is maximized along the new coordinate axes. By the affine equivariance of $T(X), S_1(X)$, $S_2(X)$ described earlier, the ICS-components $Z$ are unaffected by the coordinate system of the original data. That is, starting from $x_1, \ldots, x_n$ or $Ax_1 + b, \ldots, Ax_n + b$ yields the same $Z$. Practically, this implies that the latent components in $Z$ often reveal ``hidden'' patterns that are independent of the used coordinate system and that were not visible in the original data, see the examples in \cite{tyler2009invariant}. As a framework, ICS unifies several classical multivariate analysis methods. E.g., PCA, discriminant analysis and canonical coordinate analysis are all obtained from the above by selecting $T(X), S_1(X), S_2(X)$ appropriately \cite{tyler2009invariant}.

To perform \metodi{}, we achieve anonymization in the latent space by randomly permuting every column of $Z$ (as in SA) to obtain $Z^* \in \mathbb{R}^{n \times p}$. The result is then back-transformed to the original data scale as
\begin{align*}
    X^* = Z^* V' S_1^{1/2}(X) + 1_n T(X)'
\end{align*}
to obtain the final \metodi{}-anonymized data $X^*$.



\metodi{} can be used for data sets $X$ containing also binary or class variables (when encoded as binary variables). However, then it might happen that the corresponding variables in $X^*$ are no longer binary, in which case we discretize these columns to 0/1 according to the proportions taken from the original data.


\subsection{Parameter selection}\label{subsec:parameters}

We first note that the classical SA is obtained as a special case of \metodi{}, when $T(X) = \Bar{x}$, $S_1(X) = I_p$ and $S_2(X) = S$ where $S$ denotes the covariance matrix of $X$. In particular, the choice of $S_1(X)$ in SA is completely independent of the data, showing that SA does not use the full extent of the ICS-framework, and more elaborate results can be expected with other, less trivial choices of $S_1$.

Scatter matrices are typically divided into three groups, based on how well they tolerate outliers, see \cite{tyler1982radial} for more technical definitions:
\begin{itemize}
    \item Class I scatters break down already under a single bad enough outlier and this group includes, e.g., the covariance matrix and the covariance matrix of fourth moments. A natural choice of location vector $T$ for class I scatter is the equally non-robust mean vector $\Bar{x}$.
    \item Class II contains scatters that are ``moderately'' tolerant to outliers. That is, a single outlier cannot break a class II scatter but large groups will. This class includes, e.g., Tyler's shape matrix \cite{tyler1987distribution} and the Hettmansperger-Randles~\mbox{(H-R)} estimator \cite{hettmansperger2002practical}, the latter of which automatically produces also a location vector.
    \item Class III contains highly tolerant estimators that can continue to function even if the data contain up to 50\% outliers and contaminated observations. The classical example of such a scatter is the minimum covariance determinant (MCD) \cite{hubert2018minimum} which also produces an accompanying location vector.
\end{itemize}

Currently, a variety of scatter matrix implementations are available, for example, in the following R packages: \texttt{rrcov} \cite{rrcov}, \texttt{ICSNP} \cite{ICSNP}, \texttt{robustbase} \cite{robustbase}, and \texttt{ICS} \cite{ICS}.


\subsection{Simulation study}\label{sec:simulation}

To evaluate the robustness of \metodi{} and SA under finite data, we conducted a simulation study designed to assess the methods' ability to protect the privacy of outlying observations while maintaining the utility of the anonymized dataset for subsequent linear regression analysis. 

The simulation included $1,000$ independent replications across sample sizes $n \in \{20, 40, 120, 240, 480\}$, dimensionalities $(p + 1) \in \{4, 8, 16, 32\}$ and outlier severity $\kappa \in \{0, 2, 4, 8, 16\}$. For each replication, we generated a ``clean'' dataset $(X, y)$ where the feature matrix $X \in \mathbb{R}^{n \times p}$ and the response $y \in \mathbb{R}^n$. The response was generated via linear model:
\begin{equation*}
    y = X \beta + \epsilon, \quad \epsilon \sim \mathcal{N}(0, 1)
\end{equation*}
where the coefficients $\beta_j = (-1)^{j+1}/\sqrt{j}$ for $j \in \{1, \dots, p\}$. The feature matrix $X$ and outliers in $y$ were generated according to two distinct structures:

\paragraph{Scenario 1 (independent features, one response outlier)}
The features were sampled from $\mathcal{N}_p(0, \Sigma)$ with a covariance matrix $\Sigma = \text{diag}(p, p-1, \dots, 1)$. Only the final observation was modified into an outlier: $\tilde y_n = y_n + \kappa + \delta$, where $\delta \sim \mathcal{N}(0, 0.4^2)$. This represented a situation where a single subject's uniqueness poses an extreme re-identification risk.

\paragraph{Scenario 2 (correlated features, cluster of response outliers)}
The covariance matrix was set to $\Sigma = Q \Lambda Q'$, with $Q$ being a random orthogonal matrix. The diagonal matrix $\Lambda$ contained decaying eigenvalues generated by $\lambda_j = \text{sort}\left(\frac{p-j+1}{p} + u_j\right)$, where $u_j \sim \mathcal{U}(0,0.05)$. For a 10\% subset $O$ of the sample, we shifted the response: $\tilde y_i = y_i + \kappa + \delta_i$ for $i \in O$, where $\delta_i \sim \mathcal{N}\left(0,0.4^2\right)$. This scenario simulated structured anomalies such as batch effects or demographic sub-clusters.

For \metodi{}, we evaluated eight scatter matrix pairings $(S_1, S_2)$, where the location vector was determined by the choice of $S_1$. The configurations included: 
\begin{itemize}
    \item Class I+I: Covariance matrix and covariance matrix of fourth moments.
    \item Class II+I: H-R estimator and covariance matrix.
    \item Class II+II: H-R estimator and Tyler's shape matrix.
    \item Class III+I: MCD (50\% or 75\%) and covariance matrix.
    \item Class III+II: MCD (50\% or 75\%) and H-R estimator.
    \item Class III+III: MCD 50\% and MCD 75\%.
\end{itemize}

To evaluate the privacy, we used outlier replication error (ORE), calculated as the average normalized squared distance between each original outlier and its nearest neighbor in the anonymized data:
$$
\text{ORE} = \frac{1}{|O|} \sum_{k \in O} \min_{i = 1, \ldots, n} \left( \frac{\| \tilde{z}_{k} - z_i^* \|^2}{\| \tilde{z}_{k} \|^2} \right),
$$ where $\tilde{z}_k$ denotes an original outlying row $(x_k,y_k)$, $z_i^*$ is an anonymized row, and $O$ is the set of outlier indices. Per Theorem~\ref{theo:motivating_result}, a low ORE indicates that outliers are faithfully replicated, compromising privacy; therefore, higher values represent better protection.

To assess the preservation of utility, we measured the Euclidean distance between the true coefficients $\beta$ and the estimates $\hat{\beta}^*$ derived from the anonymized dataset using ordinary least squares: $\| \beta - \hat{\beta}^* \|$. Lower values indicate that the anonymization process better maintains the data's utility.

\subsection{Empirical data example}

To evaluate the performance of \metodi{} on real-world data, we utilized the Breast Cancer Wisconsin Diagnostic Database \cite{WBCD_data}, which contains nuclear characteristics for breast cancer diagnosis. Specifically, we employed a processed version of the dataset provided by \cite{Campos2016}, consisting of 30 numeric attributes for 367 subjects, including ten identified outliers ($2.72\%$) and 357 inliers ($97.28\%$). 

The downstream task was designed to model the final ten variables, which represent extreme cell behavior (``worst'' measurements), using the first 20 variables, which measure mean cell characteristics and heterogeneity, as predictors. This setup addressed the research question: ``Can average measurements accurately predict extreme cell characteristics?''

We modeled this relationship using Lasso regression with 10-fold \sloppy cross-validation, generating anonymized data sets 2,000 times with both \metodi{} and SA. Utility was quantified using the Euclidean distance $\| \hat\beta - \hat\beta^\ast\|$. Furthermore, we assessed the stability of variable selection by comparing the predictors selected from the anonymized data against those selected from the original dataset. This comparison utilized recall, false positive rate (FPR), precision and Jaccard index (Table~\ref{tab:metrics}). Privacy protection was quantified using ORE.

\begin{table}[ht]
\centering
\caption{Stability metrics for variable selection.}
\label{tab:metrics}
\begin{tabular}{lll}
\toprule
\textbf{Metric} & \textbf{Formula} & \textbf{Definition} \\ 
\midrule
Recall & $\frac{TP}{TP + FN}$ & \begin{tabular}[c]{@{}l@{}}The proportion of original predictors correctly \\ retained in the model fitted to anonymized data.\end{tabular} \\ 
\addlinespace
\begin{tabular}[c]{@{}l@{}}False Positive\\ Rate\end{tabular} & $\frac{FP}{FP + TN}$ & \begin{tabular}[c]{@{}l@{}}The rate at which original data zero coefficients\\ were incorrectly estimated as non-zero.\end{tabular} \\ 
\addlinespace
Precision & $\frac{TP}{TP + FP}$ & \begin{tabular}[c]{@{}l@{}}The probability that a selected predictor was a\\ true predictor from the original set.\end{tabular} \\ 
\addlinespace
\begin{tabular}[c]{@{}l@{}} Jaccard\\ Index \end{tabular} & $\frac{TP}{TP + FP + FN}$ & \begin{tabular}[c]{@{}l@{}}The intersection over union, measuring the total\\ overlap in model architecture.\end{tabular} \\
\bottomrule
\multicolumn{3}{l}{\scriptsize TP: true positive; FN: false negative; FP: false positive; TN: true negative}\
\end{tabular}
\end{table}

To quantify the privacy-utility trade-off, we defined relative privacy efficiency (RPE):
$$
\mathrm{RPE} = \frac{\sqrt{\mathrm{ORE}}}{\text{Utility metric}}
$$
where the utility metric was the Euclidean distance, FPR and $(1 - \text{metric})$ for recall, precision, and the Jaccard index. Large values of RPE are beneficial and this metric effectively captures the ``return on investment'', specifically, how much privacy is gained for every unit of utility sacrificed. 

We compared \metodi{} (utilizing MCD $50\%$ and $75\%$ scatter matrices) against SA. To facilitate a direct comparison between the two methods, we calculated a ratio of their respective RPEs. A ratio exceeding $1.0$ indicates that \metodi{} offers better privacy-utility exchange, whereas a ratio below $1.0$ favors SA. Finally, we assessed the stability of these results by calculating 95\% bootstrap confidence intervals ($B = 2,000$). An interval entirely above the $1.0$ threshold demonstrates a significant performance advantage for \metodi{} and vice versa for SA.


\section{Results}\label{sec:results}

\paragraph{Simulated data} Figure~\ref{fig:sim_results1} presents the simulation results for Scenario 1. Across all combinations of $n$ and $p$, \metodi{} employing more robust, that is, more tolerant, scatter matrices achieves greater privacy, i.e., higher outlier replication errors (consistently positioned further to the right on the x-axis). In contrast, \metodi{} utilizing at least one class I scatter matrix tends to perform worse than SA in terms of ORE (often located to the left of SA). The superior performance of \metodi{} employing class II and III scatter matrices becomes particularly pronounced as outlier severity increases.

\begin{figure}[!ht]
    \centering
    \includegraphics[width=\linewidth]{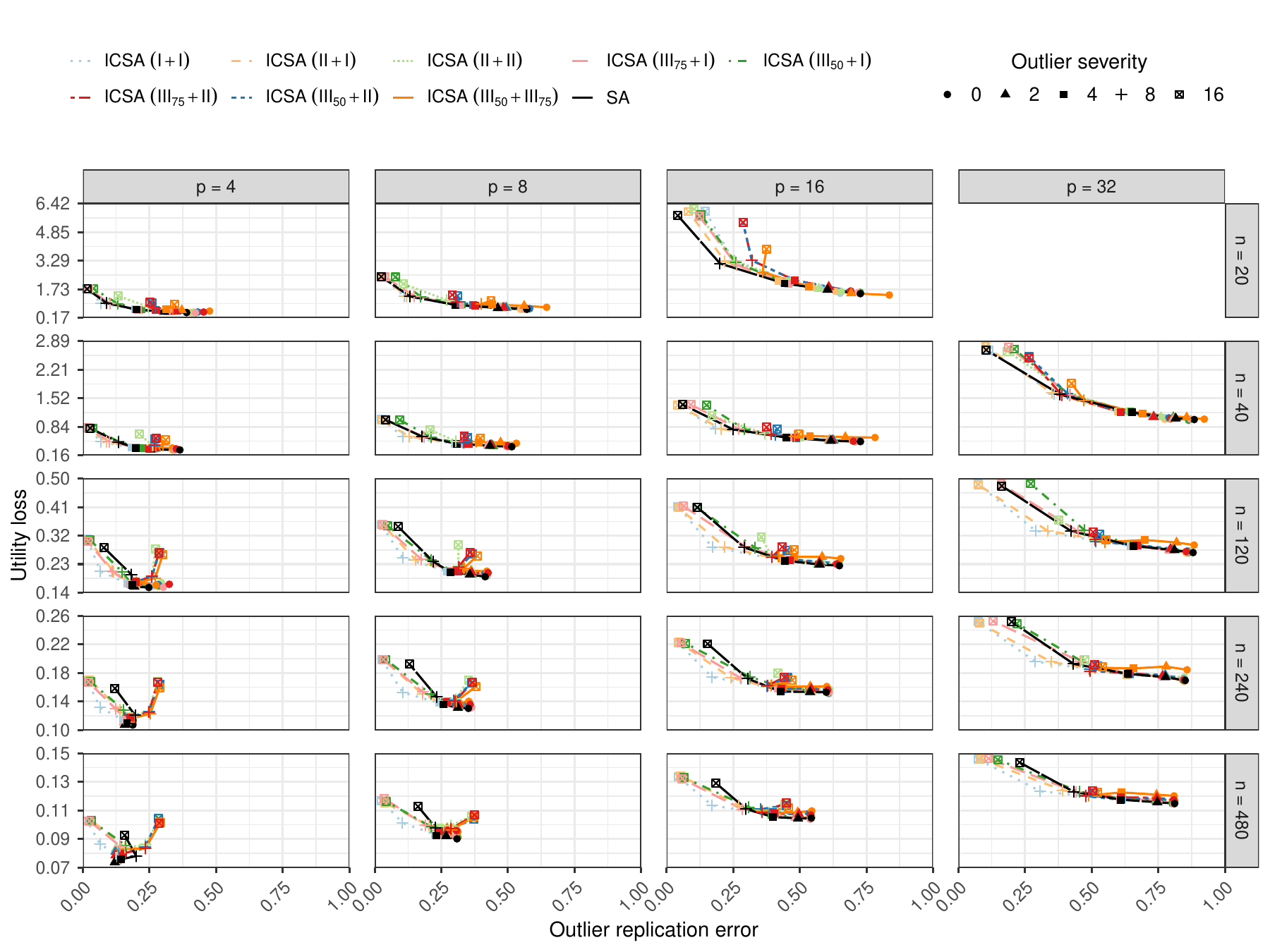}
    \caption{Utility loss versus outlier replication error in Scenario 1 across 1,000 simulations. Panels display combinations of $n$ and $p$ (varying y-axis). Lines denote methods, and point shapes indicate outlier severity (distance from the mean). Lower utility loss and higher replication error are preferred (bottom right). As outlier severity increases, methods employing more robust scatter matrices tend to exhibit superior performance.}
    \label{fig:sim_results1}
\end{figure}

In general, increasing the sample size leads to improved utility (i.e., lower utility loss), while increasing dimensionality enhances privacy (i.e., higher ORE). This behavior is expected, as in this scenario only a single response outlier is present; consequently, its influence on the overall data diminishes with increasing $n$ and $p$. Surprisingly, in many settings, \metodi{} employing class II or III scatter matrices also achieves improved utility in addition to enhanced privacy, compared to SA.

In Scenario 2 (Figure~\ref{fig:sim_results2}), both utility and privacy deteriorate in comparison to Scenario 1, with the decline in utility being more pronounced than that in privacy. This is expected since utility is evaluated via a linear regression model, which is more sensitive to widespread contamination than the privacy metric. Overall, the differences between methods are less marked. The only exception is \metodi{} employing two class III scatter matrices, which exhibits superior performance, particularly under the highest level of outlier severity.

\begin{figure}[p]
    \centering
    \includegraphics[width=\linewidth]{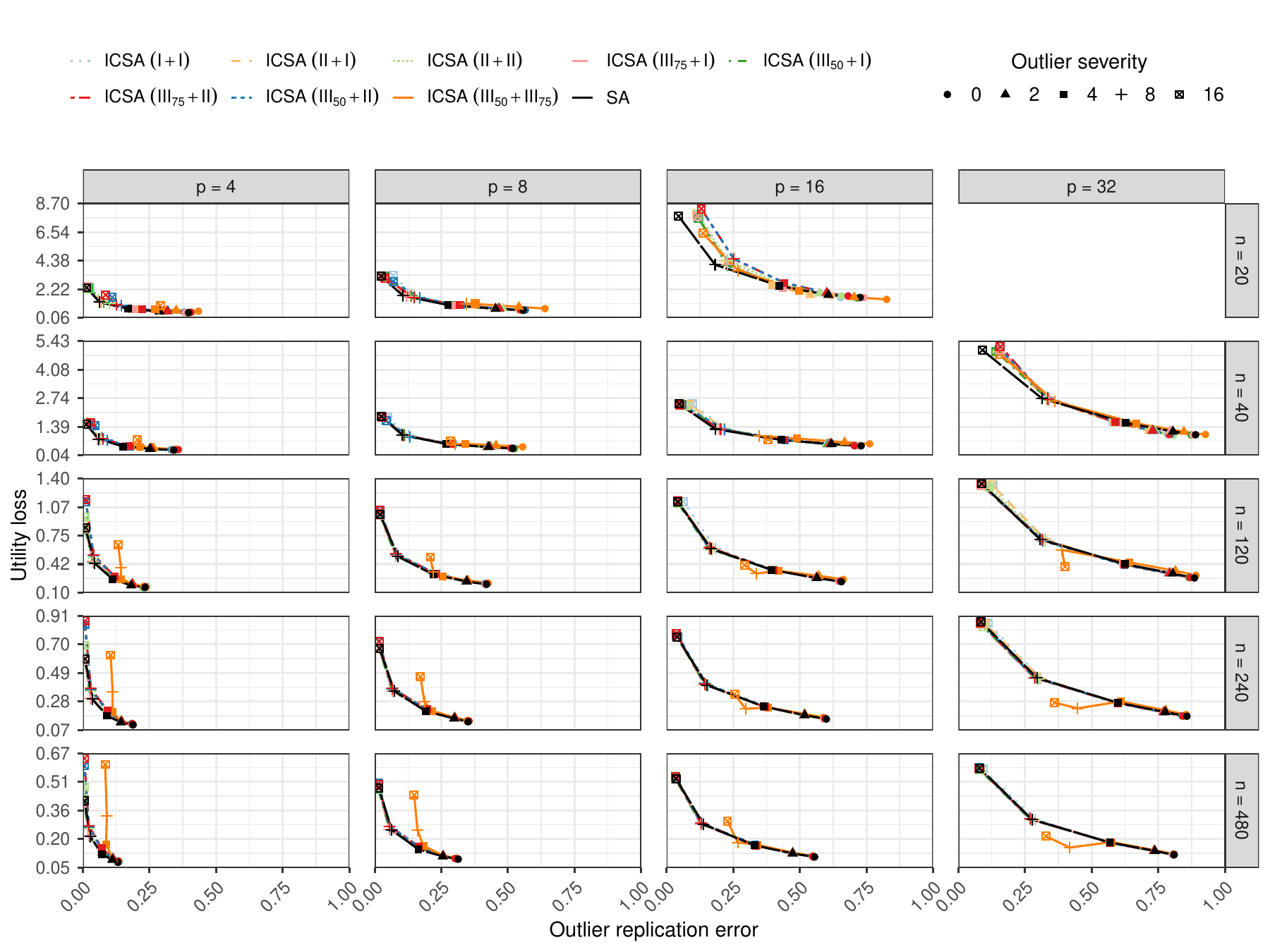}
    \caption{Utility loss versus outlier replication error in Scenario 2 across 1,000 simulations. Panels display combinations of $n$ and $p$ (varying y-axis). Lines denote methods, and point shapes indicate outlier severity (distance from the mean). Lower utility loss and higher replication error are preferred (bottom right). \metodi{} with two class III scatter matrices exhibits a clear advantage under high outlier severity.}
    \label{fig:sim_results2}
\end{figure}

\begin{table}[p]
\caption{RPE ratios between \metodi{} and SA based on real-world data, with 95\% bootstrap confidence intervals. Values above 1 indicate a more favourable privacy--utility trade-off for \metodi{}, whereas values below 1 favour SA.}
\label{tab:rwe_results}
\scriptsize
\setlength{\tabcolsep}{2pt}
\renewcommand{\arraystretch}{1.3}
\centering
\begin{tabular}{lccccc}
\toprule
\textbf{Variable}
& \textbf{Distance}
& \textbf{Recall}
& \textbf{FPR}
& \textbf{Precision}
& \textbf{Jaccard} \\
\midrule
Radius        & 2.98 (2.85, 3.11) & 3.56 (3.36, 3.77) & 4.79 (4.55, 5.05) & 4.65 (4.45, 4.86) & 4.31 (4.12, 4.50) \\
Texture       & 3.92 (3.77, 4.08) & 5.75 (5.38, 6.14) & 4.74 (4.54, 4.94) & 4.98 (4.80, 5.17) & 5.26 (5.06, 5.46) \\
Perimeter     & 3.20 (3.08, 3.33) & 2.87 (2.65, 3.10) & 3.12 (3.00, 3.25) & 3.33 (3.21, 3.46) & 3.36 (3.24, 3.48) \\
Area          & 1.89 (1.82, 1.96) & 1.12 (1.02, 1.24) & 3.58 (3.43, 3.74) & 3.33 (3.20, 3.46) & 2.97 (2.87, 3.07) \\
Smoothness    & 1.16 (1.12, 1.21) & 4.36 (4.06, 4.70) & 3.79 (3.60, 3.98) & 3.97 (3.79, 4.16) & 4.06 (3.88, 4.23) \\
Compactness   & 1.69 (1.62, 1.76) & 1.72 (1.58, 1.89) & 4.24 (4.06, 4.41) & 4.17 (4.01, 4.34) & 3.85 (3.70, 4.01) \\
Concavity     & 1.20 (1.16, 1.25) & 2.54 (2.37, 2.72) & 2.84 (2.72, 2.97) & 2.71 (2.59, 2.83) & 2.78 (2.67, 2.90) \\
Concave p.    & 2.61 (2.50, 2.72) & 2.74 (2.53, 2.98) & 6.75 (6.45, 7.06) & 6.75 (6.48, 7.02) & 6.03 (5.81, 6.26) \\
Symmetry      & 3.03 (2.91, 3.15) & 1.81 (1.69, 1.92) & 6.10 (5.80, 6.39) & 5.03 (4.81, 5.26) & 3.32 (3.19, 3.45) \\
Fractal dim.  & 0.74 (0.71, 0.78) & 1.97 (1.82, 2.14) & 5.41 (5.18, 5.64) & 5.18 (4.98, 5.39) & 4.33 (4.18, 4.51) \\
\midrule
Overall       & 2.88 (2.81, 2.95) & 2.33 (2.27, 2.39) & 3.91 (3.85, 3.97) & 3.95 (3.89, 4.01) & 3.58 (3.53, 3.63) \\
\bottomrule
\end{tabular}
\end{table}

\paragraph{Empirical data} In real-world analysis (results are reported as mean (SD)) \metodi{} provided substantially higher privacy protection than SA, with an ORE of 0.00490 (0.0044) vs. 0.00104 (0.0009). While SA recorded higher raw utility in Euclidean distance, 267.68 (759.82) vs. 436.69 (1208.80), and parameter selection, specifically recall 0.942 (0.070) vs. 0.884 (0.118); precision 0.810 (0.114) vs. 0.774 (0.105); Jaccard similarity 0.767 (0.102) vs. 0.694 (0.101); and FPR 0.529 (0.286) vs. 0.634 (0.264), these metrics feature high variance and overlapping intervals. Crucially, when evaluating the privacy--utility trade-off via relative privacy efficiency ratios (Table~\ref{tab:rwe_results}), \metodi{} outperforms SA in 49 out of 50 cases, demonstrating its definitive superiority in optimizing both objectives simultaneously.

\section{Discussion}\label{sec:discussion}

This work considered data anonymization in the presence of outliers and introduced \metodi{} as a robust alternative to spectral anonymization. Our theoretical result established that the masking mechanism of SA breaks down under strong outlying observations, providing a formal motivation for robustifying the latent-space transformation used for anonymization. By replacing the PCA decomposition of SA with invariant coordinate selection, \metodi{} performs the anonymization in a latent space that can be made resistant to contamination through the choice of scatter matrices.

The simulation study demonstrated that the privacy--utility performance of \metodi{} depends strongly on the robustness properties of the selected scatter matrices. In Scenario 1, where contamination was limited to a single outlier, \metodi{} employing more robust class II and III scatter matrices consistently achieved higher outlier replication error than SA, with differences becoming increasingly pronounced as outlier severity increased. Notably, these privacy gains were often accompanied by comparable or even improved utility, indicating that privacy protection need not be achieved at the expense of analytical usefulness. In contrast, implementations involving at least one non-robust class I scatter matrix tended to offer little advantage over SA. 

This intermediate performance of SA is also intuitive from a robustness perspective, as one of the scatter matrices implicitly used by SA is the identity matrix, which is entirely independent of the observed data and therefore unaffected by contamination. Consequently, SA inherits a degree of robustness through this fixed component, although its remaining PCA-based structure remains sensitive to outliers. This helps explain why SA was often observed between clearly non-robust and strongly robust \metodi{} implementations.

In Scenario 2, where contamination was more pervasive, both privacy and utility deteriorated for all methods, with utility being more strongly affected. This is consistent with the fact that utility was evaluated through linear regression, which is sensitive to widespread contamination. Under this more challenging setting, differences between methods were generally smaller, although \metodi{} employing two class III scatter matrices retained a clear advantage under the highest outlier severity. Taken together, these findings suggest that the relative benefit of robust anonymization is greatest when outliers are severe and when sufficiently robust scatter estimators are used.

The empirical analysis based on the Wisconsin Breast Cancer Diagnostic dataset \cite{WBCD_data} further illustrated the practical relevance of the proposed approach. Although SA achieved better mean values for several utility-oriented selection metrics, \metodi{} provided substantially stronger privacy protection. Moreover, when privacy and utility were jointly evaluated through the relative privacy efficiency criterion, \metodi{} outperformed SA in almost all variables and metrics. This highlights an important practical point: comparisons based solely on utility summaries may overlook gains in disclosure protection that are substantial when evaluated jointly. Ultimately, if outlier replication is sufficiently high to create a risk of identity disclosure, superior utility preservation alone does not render a method acceptable, as the anonymized data cannot be safely released.

The present study has several limitations. First, only two contamination scenarios and one benchmark dataset were considered, and broader empirical validation would be valuable. Second, utility was assessed through regression- and selection-based criteria; other downstream tasks such as classification, clustering, or causal estimation may yield different conclusions. Third, the performance of \metodi{} depends on the choice of scatter matrices. Although we outlined principled options based on robustness classes and existing implementations, practical guidance for data-specific selection warrants further investigation.

Future research could consider extensions of \metodi{} to high-dimensional settings, mixed data types, and nonlinear latent representations. It would also be of interest to study formal privacy guarantees under robust latent-space anonymization and to compare the proposed approach with modern privacy-preserving frameworks such as differential privacy \cite{Dwork2006_DP}. Overall, the results indicate that explicitly accounting for outliers can materially improve anonymization performance and that robust latent-space methods such as \metodi{} offer a promising direction for privacy-preserving data release.

\begin{credits}
\subsubsection{\ackname}  We would like to thank Klaus Nordhausen for his recommendations regarding the various available scatter matrices. The work of KP and JV was supported by the Research Council of Finland (grants 347501, 353769, 368494). The work of KP was supported by the Finnish Cultural Foundation (grant 00260122).

\subsubsection{\discintname}
The authors have no competing interests to declare that are relevant to the content of this article.
\end{credits}

\appendix

\section{Auxiliary lemmas and proof of Theorem \ref{theo:motivating_result}}\label{sec:proofs}

This appendix uses the following formulation of SA (for a data set of size $(n + 1) \times p$), which is completely equivalent to the PCA-based definition given in Section \ref{sec:methods}. Let a centered data matrix $X - 1_{n + 1} \Bar{x}'$ have the singular value decomposition $U D V'$. Then, the spectrally anonymized data are $X^* = U^* D V' + 1_{n + 1} \Bar{x}'$ where $U^*$ is obtained from $U$ by randomly permuting the elements of each of its columns.

\begin{lemma}\label{lem:variance_lemma}
    Let $x_1, \ldots, x_{n + 1} \in \mathbb{R}$ be such that $x_1, \ldots, x_n \in [-M, M]$ and $x_{n + 1} = H > 0$ for some fixed $M, H > 0$.
    \begin{itemize}
        \item[(i)] If $H > nM$, maximal sample variance $s^2(x)$ is achieved by taking $x_1 = \cdots = x_n = -M$, in which case
        \begin{align*}
            s^2(x) = \frac{n}{(n + 1)^2} (H + M)^2 
        \end{align*}
        \item[(ii)] The sample variance satisfies
        \begin{align*}
            s^2(x) \leq \frac{1}{n + 1} (nM^2 + H^2).
        \end{align*}
        \item[(iii)] Minimal sample variance is achieved by taking $x_1 = \cdots = x_n = M$, in which case
        \begin{align*}
            s^2(x) = \frac{n}{(n + 1)^2} ( H - M )^2.
        \end{align*}
    \end{itemize}
\end{lemma}

\begin{proof}[Lemma \ref{lem:variance_lemma}]
    Let $x := (x_1, \ldots, x_n)'$. For result (i), we note that the sample variance
    \begin{align*}
        s^2(x) = \frac{1}{n + 1} (x'x + H^2) - \left\{ \frac{1}{n + 1} (x'1_n + H) \right\}^2
    \end{align*}
    has the derivatives
    \begin{align*}
        \frac{\partial}{\partial x} s^2(x) &= \frac{2 x}{n + 1} - \frac{2}{(n + 1)^2} (x'1_n + H)  1_n \\
        \frac{\partial^2}{\partial x \partial x'} s^2(x) &= \frac{2}{n + 1} \left( I_n - \frac{1}{n} 1_n 1_n' \right) + \frac{2}{n(n + 1)^2}  1_n 1_n',
    \end{align*}
    where the matrix of second derivatives is expressed as a sum of an orthogonal projector and its complement, making it positive definite. Hence, $s^2(x)$ is convex and by the Bauer maximum principle attains its maximal value when $x_1, \ldots, x_n \in \{ -M, M \}$. Letting $s_1, \ldots, s_n$ denote the corresponding signs, the variance is maximized when $( M \sum_{i = 1}^n s_i + H )^2$
    is minimized. Since $H > nM$, this occurs when all signs are set to $-1$, in which case the variance is
    \begin{align*}
        s^2(x) =& \frac{1}{n + 1} (n M^2 + H^2) - \left\{ \frac{1}{n + 1} (-n M + H) \right\}^2 = \frac{n}{(n + 1)^2} (M + H)^2.
    \end{align*}
    For the result (ii), we observe that $(n + 1) s^2(x) \leq (nM^2 + H^2)$, from which the claim follows.

    For the result (iii), letting $m := (1/n) \sum_{i = 1}^n x_i $, we have $(n + 1) \Bar{x} = n m + H$ and the sample variance can be written as
    \begin{align*}
        (n + 1) s^2(x) =& \sum_{i = 1}^n (x_i - m)^2 + n ( m - \Bar{x} )^2 + (H - \Bar{x})^2 \\
        =& \sum_{i = 1}^n (x_i - m)^2 + \frac{n}{n + 1} ( H - m )^2 
    \end{align*}
    For the first term above, its minimal value is 0, reached when all of $x_i$ are equal, whereas the second term is minimized when $m  = M$. Both conditions are simultaneously true when all $x_1 = \cdots = x_n = M$, implying that the minimal possible variance satisfies
    \begin{align*}
        s^2(x) = \frac{n}{(n + 1)^2} ( H - M )^2,
    \end{align*}
    as claimed.
\end{proof}

\begin{lemma}\label{lem:pc_direction}
    Let $x_1, \ldots, x_{n + 1} \in \mathbb{R}^p$ be such that $\| x_1 \|, \ldots, \| x_n \| \leq M$ and $\| x_{n + 1} \| = H$ for some fixed $M, H > 0$ such that $H > (n + 2)M$. Let $u$ denote the direction of the first sample principal component. Then,
    \begin{align*}
        \left| \frac{x_{n + 1}'}{\| x_{n + 1} \|} u \right| \geq 1 -  \frac{2M}{H}.
    \end{align*}
\end{lemma}

\begin{proof}[Lemma \ref{lem:pc_direction}]
    Letting $v := x_{n + 1}/ \| x_{n + 1} \|$, the projected sample satisfies $x_1'v, \ldots x_n'v \in [-M, M]$ and $x_{n + 1}'v = H$. By Lemma \ref{lem:variance_lemma}, the variance of the data along $v$, denoted by $s^2_v(x)$ thus satisfies
    \begin{align*}
        s^2_v(x) \geq \frac{n}{(n + 1)^2} ( H - M )^2.
    \end{align*}
    Consider then an arbitrary unit-length direction $w \in \mathbb{R}^p$ satisfying, without loss of generality, $w'x_{n + 1} \geq 0$. Reasoning as above, we obtain that
    \begin{align*}
          s^2_w(x) \leq 
    \begin{cases}
      \frac{n}{(n + 1)^2} (w'x_{n + 1} + M)^2, \quad & \text{if $w'x_{n + 1} > nM$},\\
      \frac{1}{n + 1} \{ nM^2 + (w'x_{n + 1})^2 \}, \quad & \text{if $w'x_{n + 1} \leq nM$}.
    \end{cases}  
    \end{align*}
    Take now any $w$ with $w'x_{n + 1} > nM$. Then, the above implies that we have $s^2_v(x) > s^2_w(x)$ if the following holds
    \begin{align*}
        \frac{n}{(n + 1)^2} ( H - M )^2 > \frac{n}{(n + 1)^2} (w'x_{n + 1} + M)^2.
    \end{align*}
    Since both quantities under the squares are positive, this can be written as
    \begin{align*}
        1 - 2 \frac{M}{H} > w'\frac{x_{n + 1}}{\| x_{n + 1} \|}.
    \end{align*}
    Consequently, any $w$ with $w'x_{n + 1} > nM$ cannot be the direction of the leading principal component unless $w'x_{n + 1}/\| x_{n + 1} \| \geq 1 - 2M/H$. 
        
    Taking then any $w$ with $w'x_{n + 1} \leq nM$, the dominance $s^2_v(x) > s^2_w(x)$ is implied if the following inequality is true
    \begin{align*}
        \frac{n}{(n + 1)^2} ( H - M )^2 > \frac{1}{n + 1} \{ nM^2 + (w'x_{n + 1})^2 \}.
    \end{align*}
    We next show that this always holds under our conditions. Using $w'x_{n + 1} \leq nM$, we see that
    \begin{align*}
        \frac{1}{n + 1} \{ nM^2 + (w'x_{n + 1})^2 \} \leq n M^2.
    \end{align*}
    Moreover, $(n + 1)^2 M^2 < (H - M)^2$ is seen to hold thanks to our assumption that $H > (n + 2) M$. As such no $w$ with $w'x_{n + 1} \leq nM$ can be the direction of the leading principal component. The proof is now completed after noting that by symmetry the arguments apply also for vectors $w$ having $w'x_{n + 1} \leq 0$.
\end{proof}

\begin{proof}[Theorem \ref{theo:motivating_result}]
    Let $X \in \mathbb{R}^{(n + 1) \times p}$ denote the original data matrix and let $\Bar{x} = \{ 1/(n + 1) \} X 1_{n + 1}$ be the sample mean. Let $U, D, V$ and $U^*$ be as described in the beginning of Section \ref{sec:proofs} (but with sample size $n + 1$ instead).
    

    Our quantity of interest can then be written as
    \begin{alignat*}{2}
        \| x_{n + 1} - x_i^* \|^2 =& \| (x_{n + 1} - \Bar{x}) - (x_i^* - \Bar{x}) \|^2 &&= \| V D U' e_{n + 1} - V D (U^*)' e_i \|^2 \\
        =& \| D U' e_{n + 1} - D (U^*)' e_i \|^2 &&= \sum_{j = 1}^p d_j^2 (u_{(n + 1)j} - u^*_{ij})^2,
    \end{alignat*}
    where $d_1, \ldots, d_p$ are the diagonal elements of $D$. Let next $i_0$ be the index for which $u_{(n + 1)1} = u^*_{i_0 1}$. This shows that $\min_{i = 1, \ldots, n + 1} \| x_{n + 1} - x_i^* \|^2 \leq \sum_{j = 2}^p d_j^2 (u_{(n + 1)j} - u^*_{i_0 j})^2$. Since the columns of $U$ have unit lengths, the maximal value for $(u_{(n + 1)j} - u^*_{i_0 j})^2$ is 2, obtained when only two elements are non-zero and equal $1/\sqrt{2}$ and $-1/\sqrt{2}$. Hence,
    \begin{align}\label{eq:upper_bound_1}
        \min_{i = 1, \ldots, n + 1} \| x_{n + 1} - x_i^* \|^2 \leq 2 (p - 1) d_2^2, 
    \end{align}
    where $d_2^2$ is the variance of the second principal component.
    
    We next obtain a bound for $d_2^2$. Let $v_1, \ldots, v_p$ denote the columns of $V$ (i.e. the directions of the principal components). By writing
    \begin{align*}
        \frac{x_{n + 1}}{\| x_{n + 1} \|} =  v_1 v_1' \frac{x_{n + 1}}{\| x_{n + 1} \|} + \cdots + v_p v_p' \frac{x_{n + 1}}{\| x_{n + 1} \|},
    \end{align*}
    we obtain that 
    \begin{align*}
        1 = \left( v_1' \frac{x_{n + 1}}{\| x_{n + 1} \|} \right)^2 + \cdots + \left( v_p' \frac{x_{n + 1}}{\| x_{n + 1} \|} \right)^2.
    \end{align*}
    Now, combining this with Lemma \ref{lem:pc_direction}, which shows, for $\| x_{n + 1} \|/M$ large enough, that the first principal component direction $v_1$ satisfies
    \begin{align*}
        \left( \frac{x_{n + 1}'}{\| x_{n + 1} \|} v_1 \right)^2 \geq \left( 1 -  \frac{2M}{\| x_{n + 1} \|} \right)^2,
    \end{align*}
    we see that
    \begin{align*}
        \left( v_2' \frac{x_{n + 1}}{\| x_{n + 1} \|} \right)^2 \leq 1 - \left( \frac{x_{n + 1}'}{\| x_{n + 1} \|} v_1 \right)^2 \leq 1 - \left( 1 -  \frac{2M}{\| x_{n + 1} \|} \right)^2,
    \end{align*}
    which implies that the projection of the outlying point satisfies $( v_2' x_{n + 1} )^2 \leq 4 M ( \| x_{n + 1} \| - M )$.
    Meanwhile, the projections of the non-outlying points satisfy $| v_2' x_1 |, \ldots , | v_2' x_n | \leq M$. Lemma \ref{lem:variance_lemma} shows that the variance of this projection is at most
    \begin{align*}
        \frac{1}{n + 1} \{ nM^2 + ( v_2' x_{n + 1} )^2 \} \leq \frac{M}{n + 1} \{ nM + 4 ( \| x_{n + 1} \| - M ) \},
    \end{align*}
    which is then also an upper bound for the variance $d_2^2/(n + 1)$ of the second PC. Hence, \eqref{eq:upper_bound_1} gives
    \begin{align}\label{eq:upper_bound_2}
        \min_{i = 1, \ldots, n + 1} \| x_{n + 1} - x_i^* \|^2 \leq 2 (p - 1) M \{ (n - 4) M + 4 \| x_{n + 1} \| \},
    \end{align}
    from which we obtain that
    \begin{align*}
        \min_{i = 1, \ldots, n + 1} \frac{\| x_{n + 1} - x_i^* \|^2}{\| x_{n + 1} \|^2} = \mathcal{O} \left( \frac{M^2}{\| x_{n + 1} \|^2} \right) + \mathcal{O} \left( \frac{M}{\| x_{n + 1} \|} \right),
    \end{align*}
    where the second term dominates. Finally, since the upper bound \eqref{eq:upper_bound_2} is valid for all possible spectral anonymizations, we obtain the desired claim.
\end{proof}


%
%
%
\bibliographystyle{splncs04}
\bibliography{mybibliography}

\end{document}